\documentclass{article}
\usepackage{graphicx} 
\usepackage[margin=1in]{geometry}
\usepackage{amsmath, amsthm, amssymb}
\usepackage{hyperref}
\usepackage{multicol}
\usepackage[title]{appendix}
\hypersetup{colorlinks=true, pdfstartview=FitV, linkcolor=blue, citecolor=blue, urlcolor=blue}
\usepackage{times}
\usepackage{comment}
\usepackage[usenames,dvipsnames]{color}
\usepackage{bbm}
\usepackage{subcaption}
\usepackage{mathtools}
\usepackage{pdfpages}
\usepackage{bm}
\usepackage{verbatim}
\usepackage{mathabx}
\usepackage{geometry}
\usepackage{enumerate}

\numberwithin{equation}{section}

\newtheorem{thm}{Theorem}[section]
\newtheorem{Prop}[thm]{Proposition}
\newtheorem{lem}[thm]{Lemma}

 \newtheorem{Rmk}[thm]{Remark}
 \newtheorem{assumptions}[thm]{Assumptions}
 \newtheorem{question}[thm]{Problem}
 \newtheorem{Conj}[thm]{Conjecture}
\newcommand{\ba}{\begin{aligned}}
\newcommand{\ea}{\end{aligned}}

\title{Stroboscopic Ergodic Averages\\ and Radial Symmetry in the Vortex Model}
\author{Emanuele Caglioti and Marco Cecchini \\
        \small Dipartimento di Matematica Guido Castelnuovo\\ \small Sapienza Università di Roma}
\date{}
\begin{document}
\maketitle
\begin{abstract}
We study the point-vortex model in the plane in the case of $N$ vortices of equal 
positive circulation, and we investigate the long-time behaviour of the empirical 
vorticity measure. We conjecture that, for almost every initial datum, the time 
average of the empirical measure exists and is invariant under rotations about the 
centre of vorticity.

We prove the conjecture for $N=3$ by exploiting the integrability of the system. 
For $N\ge 4$ we reduce the problem to a question about a measure-preserving 
dynamical system on the shape space: after replacing time with a suitably chosen 
rotation angle as independent variable, radiality is equivalent to the independence 
of stroboscopic ergodic averages from the initial phase along a single trajectory. 
We formulate this reduction as a general ergodic problem, discuss its scope through 
elementary counterexamples, and verify the required hypotheses in explicit regions 
of the phase space for $N=3$ and $N=4$.
\end{abstract}

\section{Introduction}\label{sec:1}

The point-vortex model \cite{H,K} is a Hamiltonian system describing the interaction 
of $N$ singular vortices in the plane. Throughout this paper we consider the case 
of $N$ vortices of equal circulation $\omega = \Omega/N$, and we denote their 
positions by $z_1,\dots,z_N\in\mathbb{R}^2$, with $z_i=(x_i,y_i)$. 
The equations of motion are
\begin{equation}\label{eq:hamilton}
\dot{z}_i = -\nabla_{z_i}^{\perp} H, \qquad i=1,\dots,N,
\end{equation}
with Hamiltonian
\begin{equation}\label{Ham}
H = -\frac{\Omega}{4\pi N}\sum_{i\neq j}\log|z_i-z_j|,
\end{equation}
which, written in components, gives
\begin{equation}\label{VM}
\begin{cases}
\dot{x}_i = -\dfrac{\Omega}{2\pi N}\displaystyle\sum_{j\neq i}
   \dfrac{y_i-y_j}{|z_i-z_j|^2}, \\[1.2ex]
\dot{y}_i = +\dfrac{\Omega}{2\pi N}\displaystyle\sum_{j\neq i}
   \dfrac{x_i-x_j}{|z_i-z_j|^2}.
\end{cases}
\end{equation}
System \eqref{VM} admits four independent first integrals: the two components of 
the centre of vorticity $Z$, the moment of inertia $I$, and the Hamiltonian $H$, 
where
\begin{equation}\label{eq:firstintegrals}
Z = \frac{1}{N}\sum_i z_i, \qquad I = \frac{1}{N}\sum_i |z_i|^2.
\end{equation}

We are interested in the long-time behaviour of the empirical vorticity measure
\[
\omega_N(t) \;=\; \frac{1}{N}\sum_{i=1}^N \delta_{z_i(t)},
\]
and specifically in whether its time average is invariant under rotations about 
the centre of vorticity. Our main conjecture is the following.

\begin{Conj}\label{congetturan}
Consider the vortex model for $N$ equal vortices $z=(z_1,\dots,z_N)$, each of unit 
circulation, on the pre-shape set
\[
\Gamma_N \;=\; \left\{\, z\in(\mathbb{R}^2)^N \;:\; 
\sum_{k=1}^N z_k = 0,\; \frac{1}{N}\sum_{k=1}^N |z_k|^2 = 1 \,\right\}.
\]
Then, for almost every $z\in\Gamma_N$, the time average of the empirical vorticity 
measure exists in the sense of weak convergence of measures,
\begin{equation}\label{eq:mudef}
\mu(dz)\;=\;\lim_{T\to\infty}\frac{1}{T}\int_0^T\!\!dt\;
\frac{1}{N}\sum_{i=1}^N \delta_{z-z_i(t)},
\end{equation}
and $\mu$ is invariant under rotations about the origin. Equivalently, for every 
smooth observable $\Phi:\mathbb{R}^2\to\mathbb{R}$ and every $\theta\in[0,2\pi)$, 
the time average
\[
\langle \Psi_\theta \rangle \;=\; \lim_{T\to\infty}\frac{1}{T}\int_0^T\!\!dt\;
\Psi_\theta(z(t))
\qquad\text{of}\qquad
\Psi_\theta(z) \;\equiv\; \frac{1}{N}\sum_{k=1}^N \Phi\circ R_\theta(z_k)
\]
exists and does not depend on $\theta$.
\end{Conj}

The heuristic ground for this conjecture is that the system is invariant under 
rotations about the origin, and it is hard to see how the shape of the 
configuration could remain coherently phase-locked to the rotation angle for a 
positive-measure set of initial data.

For $N=3$ the system is integrable. The triangle formed by the vortices evolves 
periodically for almost every initial datum, with a period $T$ that depends on 
the orbit; during one period the triangle rotates by an angle $\Delta\Theta$, 
which turns out to be an irrational multiple of $2\pi$ for a.e.\ initial datum. 
Radiality of the time average follows.

For $N\ge 4$ the system is generically non-integrable, yet the same picture is 
expected: in KAM regions the motion lies on invariant tori with frequencies that 
generically vary with the initial datum, so they are typically incommensurable 
with the rotation frequency; in chaotic regions, decay of time correlations 
should make the shape decorrelate from the rotation angle, again yielding 
radiality. We are not able to prove the conjecture in the general case.

We substantiate the heuristics above by transforming the problem into an ergodic 
statement about stroboscopic averages along a single trajectory of a 
measure-preserving system on the shape space. Concretely, after a 
reparameterisation that replaces time $t$ with a geometric rotation angle 
$\theta$, and under mild assumptions on the choice of $\theta$, the conjecture 
reduces to the following claim: for almost every trajectory, the stroboscopic 
averages of a smooth observable at angles $\alpha + 2\pi k$, $k\in\mathbb{N}$, 
are independent of the initial phase $\alpha$ (Problem \ref{ergodic-question}).

\medskip
The paper is organised as follows. In Section~\ref{sec:2} we give a brief historical 
survey of the point-vortex problem. In Section~\ref{sec:3} we prove the conjecture 
for $N=3$. In Section~\ref{sec:4} we introduce the shape-angle reparameterisation 
and formulate the ergodic problem to which the conjecture reduces. In Section~\ref{sec:5} 
we apply this framework to $N=3$, to the KAM regime for $N=4$, and conclude with an 
analytical proof of the conjecture on the symmetric four-vortex submanifold (Remark~\ref{rmk:sym4}).

\section{A brief survey of point-vortex dynamics}\label{sec:2}

The mathematical study of discrete vortices begins with Helmholtz \cite{H}, who 
showed that in an inviscid incompressible fluid vortex lines are transported by 
the flow, and with Kirchhoff \cite{K}, who cast the dynamics of $N$ point vortices 
in the plane as a Hamiltonian system.

The three-vortex problem was first analysed in full detail by Gröbli 
\cite{Gro77}, who established its integrability by explicit quadratures; 
Poincaré \cite{Poi1893} later gave the Hamiltonian proof, showing that the 
system admits three independent conserved quantities in involution, namely the 
Hamiltonian $H$, the centre of vorticity $Z$ and the moment of inertia $I$. A 
historical account of Gröbli's work and its later developments can be found in 
the review by Aref, Rott and Thomann \cite{ART92}.

Gröbli's key idea was to reduce the problem to a closed system for the squared 
inter-vortex distances $s_{jk}=|z_j-z_k|^2$, which describe the evolution of the 
triangle formed by the vortices independently of its position and orientation. 
In modern language this is the motion on the \emph{shape sphere}. A global 
phase-sphere picture was later developed by Novikov \cite{Nov1975} and Aref 
\cite{Aref79}. Jacobi-type coordinates have long been the standard tool for 
decoupling the centre of vorticity, but they typically break the permutation 
symmetry and introduce coordinate singularities; a recent globally non-singular 
canonical reduction, combining Jacobi coordinates with Nambu brackets, was 
proposed by Anurag, Goodman and O'Grady \cite{Anurag2024}, and is the one we 
shall use in Section \ref{sec:3}.

For $N\ge 4$ the picture changes drastically. Ziglin \cite{Z} proved that the 
four-vortex problem is not integrable in general, while Khanin \cite{Kh}, using 
KAM theory, exhibited a positive-measure set of initial data on which the motion 
is quasi-periodic; see also Celletti and Falcolini \cite{CF} and Lim \cite{L91} 
for related KAM results, and Chen and Li \cite{CL23} for explicit 
non-degeneracy conditions in the planar $N$-vortex system. Numerical evidence 
for chaotic advection in the four-vortex problem was already reported by Aref 
and Pomphrey \cite{AP82}.

This is only a very partial overview; we refer the reader to Aref \cite{Aref07} 
and to the bibliography by Meleshko and Aref \cite{MA} for a much more complete 
picture. For the connection with the two-dimensional Euler equation we refer to 
Marchioro and Pulvirenti \cite{MP} and to Majda and Bertozzi \cite{MB}.

\section{The three-vortex case}\label{sec:3}

In this section we prove Conjecture \ref{congetturan} for $N=3$.

\begin{thm}\label{thm}
Let $z=(z_1,z_2,z_3)\in(\mathbb{R}^2)^3$ and consider the pre-shape set
\[
S_1 \;=\; \left\{\,(z_1,z_2,z_3)\;:\;\sum_{k=1}^{3}z_k=0,\;
\sum_{k=1}^{3}|z_k|^2=1\,\right\}.
\]
Then:
\begin{enumerate}
\item for almost every $z\in S_1$ the time average
\begin{equation}\label{eq:mu3}
\mu(dz)\;=\;\lim_{T\to\infty}\frac{1}{T}\int_0^T\!\!dt\;
\frac{1}{3}\sum_{i=1}^3 \delta_{z-z_i(t)}
\end{equation}
exists in the sense of weak convergence of measures, and $\mu$ is invariant 
under rotations about the origin;
\item there exists a set of initial data (of zero Lebesgue measure on $S_1$) for 
which the time average is not rotationally invariant.
\end{enumerate}
\end{thm}

The three-vortex system is integrable because the centre of vorticity and the 
moment of inertia are conserved. We fix $(I,\Omega)=(1/3,6\pi)$, which is not restrictive since 
rescaling $(I,\Omega)$ only rescales time. On $S_1$ one has $I=1/3$; choosing $\Omega=6\pi$ gives $\Omega/(2\pi N)=1$. These space-time rescalings preserve the shape dynamics and radiality. In Remark~5.1 we instead work directly on $\Gamma_4$, where $I=1$. Among the several coordinate systems available 
for this problem, we adopt the canonical reduction of Anurag, Goodman and 
O'Grady \cite{Anurag2024}.

\subsection*{Strategy of the proof}

For almost every initial datum the shape of the triangle evolves along a 
periodic orbit of some period $T$; during one period the triangle rotates about 
the origin by a total angle $\Delta\Theta$. Radiality of the time average 
follows once one knows that $\Delta\Theta$ is an irrational multiple of $2\pi$ 
for a.e.\ initial datum. Since $\Delta\Theta$ turns out to be a real-analytic, 
non-constant function of a one-dimensional orbit parameter, this in turn follows 
from a general elementary lemma (Lemma \ref{lem:generic-irrational} below).

The reduced dynamics on the shape sphere has two families of periodic orbits, 
separated by heteroclinic connections; we treat the two families in a unified 
way after setting up the coordinates.

\subsection*{Shape coordinates and reduced dynamics}

Let $\zeta_i\in\mathbb{R}^2$ denote the position of the $i$-th vortex, and 
assume $\zeta_1+\zeta_2+\zeta_3=0$. Jacobi coordinates
\begin{equation}\label{coord0}
R_1=\zeta_1-\zeta_2, \qquad 
R_2=\frac{\zeta_1+\zeta_2}{2}-\zeta_3, \qquad 
R_3=\frac{\zeta_1+\zeta_2+\zeta_3}{3}=0,
\end{equation}
rescaled as $\mathcal{R}_1=\tfrac{1}{\sqrt{2}}R_1$, 
$\mathcal{R}_2=\sqrt{\tfrac{2}{3}}R_2$, are inverted by
\begin{equation}\label{coord1}
\zeta_1=\frac{R_1}{2}+\frac{R_2}{3}, \qquad 
\zeta_2=-\frac{R_1}{2}+\frac{R_2}{3}, \qquad 
\zeta_3=-\frac{2R_2}{3}.
\end{equation}
Following \cite{Anurag2024}, we introduce the three shape variables
\begin{equation}\label{eq:XYZ}
Z=|\mathcal{R}_1|^2-|\mathcal{R}_2|^2, \qquad
X=2\,\mathcal{R}_1\!\cdot\!\mathcal{R}_2, \qquad
Y=2\,\mathcal{R}_1\!\cdot\!\mathcal{R}_2^{\perp},
\end{equation}
which encode the relative positions of the three vortices but not the overall 
orientation. A direct computation gives 
$X^2+Y^2=4|\mathcal{R}_1|^2|\mathcal{R}_2|^2$ and hence
\begin{equation}\label{sfera}
X^2+Y^2+Z^2 \;=\; \big(|\mathcal{R}_1|^2+|\mathcal{R}_2|^2\big)^2 
\;=\; (3I)^2 \;=\; 1,
\end{equation}
so the reduced phase space is the shape sphere $\mathbb{S}^2$, projected onto 
the disk $X^2+Z^2\le 1$ in the $(X,Z)$-plane.

In these variables the equations of motion read \cite{Anurag2024}
\begin{equation}\label{XYZdot}
\dot X = 4Y\,\frac{\partial H}{\partial Z}, \qquad
\dot Y = 4Z\,\frac{\partial H}{\partial X}-4X\,\frac{\partial H}{\partial Z}, \qquad
\dot Z = -4Y\,\frac{\partial H}{\partial X},
\end{equation}
with Hamiltonian
\begin{equation}\label{HandW}
H = -\tfrac{1}{2}\log W(X,Z), \qquad
W(X,Z) = (1+Z)\!\left[\!\left(1-\tfrac{Z}{2}\right)^{\!2}-\tfrac{3}{4}X^2\right],
\end{equation}
which depends only on $X$ and $Z$. Combining \eqref{XYZdot} and \eqref{HandW} 
one finds
\begin{equation}\label{Zdot}
\dot Z \;=\; -\,\frac{3\,XY(1+Z)}{W}.
\end{equation}

The level curves of $W$ in the $(X,Z)$-plane are shown in Figure \ref{fig:XZ}; 
the black circle $X^2+Z^2=1$ corresponds to $Y=0$. Except on a set of zero 
measure (the heteroclinic level sets), all orbits of \eqref{XYZdot} are 
periodic. Two types of periodic orbit occur:
\begin{itemize}
\item[(I)] orbits entirely contained in the open disk $X^2+Z^2<1$, i.e.\ orbits 
along which $Y$ never vanishes;
\item[(II)] orbits that meet the circle $X^2+Z^2=1$, i.e.\ orbits along which 
$Y$ changes sign.
\end{itemize}
These two families are separated by the six heteroclinic trajectories 
connecting the three equilibria on $Y=0$ (Figure \ref{fig:3}).

\begin{figure}[!h]
\centerline{\includegraphics[scale=.5]{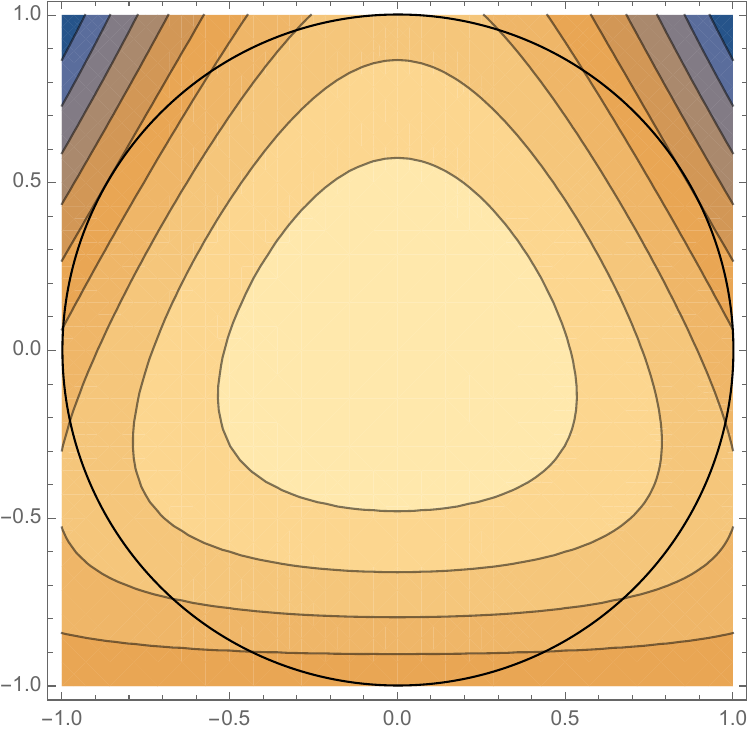}}
\caption{Level sets of $W$ in the $(X,Z)$-plane. The disk $X^2+Z^2\le 1$ is 
the projection of the shape sphere; the boundary circle corresponds to $Y=0$.}
\label{fig:XZ}
\end{figure}

\subsection*{Total rotation during one shape period}

Writing each vortex in polar form $(x_i,y_i)=(\rho_i\cos\Theta_i,\rho_i\sin\Theta_i)$, 
we define the mean rotation angle
\[
\Theta \;=\; \frac{\Theta_1+\Theta_2+\Theta_3}{3},
\]
which is the natural notion of overall rotation compatible with the 
$\mathbb{Z}_3$ permutation symmetry of the three equal vortices. When the shape 
$(X,Y,Z)$ completes one full cycle, the configuration returns to its initial 
position up to a rotation by an angle $\Delta\Theta$; if $\Delta\Theta/2\pi$ is 
irrational, the closure of the orbit is a two-torus and the empirical density 
averages to a rotationally invariant measure.

Using \eqref{coord1} and the identities
\begin{equation}\label{coord3}
R_1^2=Z+1, \qquad 
R_2^2=\tfrac{3}{4}(1-Z), \qquad 
R_1\!\cdot\!R_2=\tfrac{\sqrt{3}}{2}X, \qquad 
R_1\!\cdot\!R_2^{\perp}=\tfrac{\sqrt{3}}{2}Y,
\end{equation}
a direct computation gives $\dot\Theta_1=f(X,Z)$ with an explicit rational 
function $f$. Averaging $f$ over the $\mathbb{Z}_3$ action (a rotation by 
$2\pi/3$ in the $(X,Z)$-plane corresponds to a cyclic permutation of the vortex 
labels \cite{Anurag2024}) yields the compact expression
\begin{equation}\label{Theta}
\dot\Theta \;=\; -\,\frac{12\,(X^2+Z^2-1)(X^2+Z^2-4)}
{(\sqrt{3}X+Z+2)(\sqrt{3}X+Z-2)(\sqrt{3}X-Z+2)(\sqrt{3}X-Z-2)(Z-1)(Z+1)}.
\end{equation}

Two properties of this expression are essential in what follows.

\begin{Rmk}\label{Rmk1}
The rational function $\dot\Theta$ in \eqref{Theta} is regular outside the six 
planes
\[
X=\pm\tfrac{Z\pm 2}{\sqrt{3}}, \qquad Z=\pm 1,
\]
which are tangent to the shape sphere at the points 
$(\pm\tfrac{\sqrt{3}}{2},0,\pm\tfrac{1}{2})$ and $(0,0,\pm 1)$: three equilibria 
and three singularities of the reduced dynamics \cite{Anurag2024}. In particular, 
no factor in the denominator can change sign along an orbit. The numerator is 
non-positive and vanishes exactly on $X^2+Z^2=1$, i.e.\ on $Y=0$.
\end{Rmk}

\begin{Rmk}\label{Rmk2}
The singularities are never reached by any orbit; the equilibria can only be 
approached asymptotically. At the equilibrium $(0,0,1)$, inserting 
$X^2=\tfrac{4}{3}\big[(1-Z/2)^2-\tfrac{1}{2(Z+1)}\big]\ (W=1/2)$ (obtained from 
\eqref{HandW}) into \eqref{Theta} gives
\[
\dot\Theta \;=\; \frac{10+12Z-4Z^3}{3+3Z} \;\longrightarrow\; 3
\qquad \text{as } Z\to 1.
\]
Analogous limits hold at the other equilibria. Combined with the divergence of 
the period near heteroclinic connections, this implies that 
$\int_0^T \dot\Theta\,dt \to +\infty$ as the orbit approaches a heteroclinic 
level.
\end{Rmk}

\begin{figure}[!h]
\centering
\includegraphics[width=0.5\linewidth]{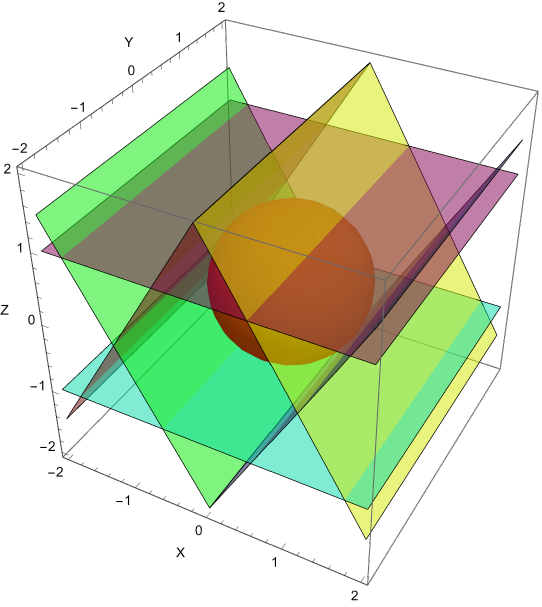}
\caption{The shape sphere and the six critical planes of Remark \ref{Rmk1}.}
\label{fig:2}
\end{figure}

\subsection*{A generic-irrationality lemma}

We isolate the elementary observation that drives both cases below.

\begin{lem}\label{lem:generic-irrational}
Let $J\subset\mathbb{R}$ be an open interval and let $F:J\to\mathbb{R}$ be a 
real-analytic function which is not constant. Then the set 
$\{s\in J:F(s)/2\pi\in\mathbb{Q}\}$ has Lebesgue measure zero; in particular, 
$F(s)$ is an irrational multiple of $2\pi$ for a.e.\ $s\in J$.
\end{lem}

\begin{proof}
For every $q\in\mathbb{Q}$ the level set $\{s\in J:F(s)=2\pi q\}$ is either $J$ 
itself or discrete, by analyticity of $F$. Since $F$ is not constant the first 
alternative is excluded, so each such set is countable; the union over 
$\mathbb{Q}$ is a countable union of countable sets, hence of Lebesgue measure 
zero.
\end{proof}

\subsection*{Region I: orbits with $Y\ne 0$}

Along a region-I orbit one has $Y\ne 0$, so by \eqref{Zdot} the extrema of $Z$ 
are attained where $X=0$. We label such an orbit by 
$R:=\max_{\mathrm{orbit}} Z$, so that $R\in(0,1)$ and, by \eqref{HandW},
\begin{equation}\label{wofR}
w \;=\; W\big|_{\mathrm{orbit}} \;=\; (1+R)\!\left(1-\tfrac{R}{2}\right)^{\!2}.
\end{equation}
The variable $Z$ then oscillates between the two roots of the cubic 
$(1+Z)(1-Z/2)^2=w$ that lie in $[-1,1]$, namely $Z_2=R$ and
\[
Z_1 \;=\; \tfrac{1}{2}\big(3-R-\sqrt{3}\sqrt{3+2R-R^2}\big),
\]
the third root $Z_3=\tfrac{1}{2}\big(3-R+\sqrt{3}\sqrt{3+2R-R^2}\big)$ lying 
outside $[-1,1]$. Solving \eqref{HandW} and \eqref{sfera} for $X$ and $Y$ as 
functions of $Z$ and substituting into \eqref{Zdot} gives
\begin{equation}\label{zetapunto}
\dot Z \;=\; \frac{1}{w}\sqrt{(Z-Z_1)(Z_2-Z)(Z_3-Z)\big(3+3Z-4Z^3+Z_1Z_2Z_3\big)},
\end{equation}
which vanishes on the orbit only at $Z=Z_1$ and $Z=Z_2$. Inserting the same 
expressions into \eqref{Theta} yields $\dot\Theta$ as an explicit rational 
function of $R$ and $Z$; combining with \eqref{zetapunto} and \eqref{wofR}, the 
total rotation over one shape period is
\begin{equation}\label{DeltaTheta1}
\Delta\Theta(R) \;=\; 2\int_{Z_1}^{Z_2}\frac{\dot\Theta}{\dot Z}\,dZ
\;=\; 2\int_{Z_1}^{Z_2}\frac{f_{\mathrm{I}}(R,Z)}{\sqrt{(Z-Z_1)(Z_2-Z)}}\,dZ,
\end{equation}
with
\[
\resizebox{\textwidth}{!}{$\displaystyle
f_{\mathrm{I}}(R,Z) \;=\; -\,\frac{4(1+R)\big(1-R/2\big)^2
(3-3R^2+R^3+3Z-4Z^3)(12-3R^2+R^3+12Z-4Z^3)}
{3\sqrt{(Z_3-Z)(3+3Z-4Z^3+Z_1Z_2Z_3)}\,(-2+R)^2(1+R)(-1+Z^2)(4-3R^2+R^3+8Z+8Z^2)}.
$}
\]
The trigonometric substitution $Z=\tfrac{1}{2}\big(Z_1+Z_2+(Z_2-Z_1)\sin\phi\big)$ 
removes the endpoint singularities of the integrand and gives
\begin{equation}\label{dthetafinal}
\Delta\Theta(R) \;=\; 2\int_{-\pi/2}^{\pi/2} 
f_{\mathrm{I}}\!\left(R,\tfrac{Z_1+Z_2+(Z_2-Z_1)\sin\phi}{2}\right)\!d\phi.
\end{equation}

\begin{Prop}\label{prop:analyticI}
The function $R\mapsto\Delta\Theta(R)$ is real-analytic on $(0,1)$ and satisfies
$\lim_{R\to 1^-}\Delta\Theta(R)=+\infty$.
\end{Prop}

\begin{proof}
Analyticity: on the integration domain the arguments of both square roots in 
$f_{\mathrm{I}}$ are strictly positive (they vanish only at the endpoints, which 
are handled by the trigonometric substitution), so the integrand in 
\eqref{dthetafinal} is jointly analytic in $(R,\phi)$; the integral over the 
compact interval $[-\pi/2,\pi/2]$ is therefore analytic in $R$.

Divergence: as $R\to 1^-$ the orbit approaches one of the six heteroclinic 
trajectories connecting the equilibria in $Y=0$ \cite{Anurag2024}. The time spent 
near the equilibrium along such a nearby orbit diverges, and by Remark 
\ref{Rmk2} $\dot\Theta$ tends to a strictly positive constant at the 
equilibrium; hence $\int_0^T\dot\Theta\,dt\to+\infty$.
\end{proof}

\subsection*{Region II: orbits with $Y$ changing sign}

Region-II orbits are confined to one of the three sub-regions of the shape 
sphere bounded by two heteroclinic connections. By the $\mathbb{Z}_3$ symmetry, 
it is enough to analyse the sub-region between the equilibria $(0,0,1)$ and 
$(\tfrac{\sqrt{3}}{2},0,-\tfrac{1}{2})$. In this region 
$\dot Z=0\Leftrightarrow Y=0$ by \eqref{Zdot}, and each orbit meets $Y=0$ in 
exactly two points, symmetric with respect to the singularity 
$(\tfrac{\sqrt{3}}{2},0,\tfrac{1}{2})$ \cite{Anurag2024}.

We parametrise the orbits in this sub-region by $\alpha\in(0,\pi/3)$: the orbit 
labelled by $\alpha$ passes through the point 
$p_\alpha=(\sin\alpha,0,\cos\alpha)$, i.e.\ the image of $(0,0,1)$ under the 
rotation by $\alpha$ around the $Y$-axis. A direct calculation shows
\[
W(p_\alpha) \;=\; W(p_{2\pi/3-\alpha}) \;=\; \tfrac{1+\cos(3\alpha)}{4},
\]
so the orbit labelled by $\alpha$ also passes through $p_{2\pi/3-\alpha}$. 
Setting
\[
Z_{1,\alpha}=\cos\alpha, \qquad
Z_{2,\alpha}=\cos\!\big(\tfrac{2\pi}{3}-\alpha\big), \qquad
Z_{3,\alpha}=\cos\!\big(\tfrac{2\pi}{3}+\alpha\big),
\]
one obtains from \eqref{HandW} and \eqref{sfera} the closed expressions
\begin{equation}\label{xII}
X \;=\; +\sqrt{\tfrac{Z^3-3Z^2+3-\cos(3\alpha)}{3+3Z}}, \qquad
Y \;=\; -2\sqrt{\tfrac{(Z_{1,\alpha}-Z)(Z-Z_{2,\alpha})(Z-Z_{3,\alpha})}{3+3Z}},
\end{equation}
valid for $Z\in[Z_{2,\alpha},Z_{1,\alpha}]$, where the factorisation of the 
numerator of $Y^2$ uses the two known roots $Z=Z_{1,\alpha},Z_{2,\alpha}$. 
Substituting into \eqref{Zdot} gives
\begin{equation}\label{Zdot2}
\begin{aligned}
\dot Z \;=\;& h(Z,\alpha)\sqrt{(Z_{1,\alpha}-Z)(Z-Z_{2,\alpha})},\\
h(Z,\alpha) \;=\;& +\frac{8}{1+\cos(3\alpha)}
\sqrt{(Z^3-3Z^2+3-\cos(3\alpha))(Z-Z_{3,\alpha})},
\end{aligned}
\end{equation}
with $h$ nowhere vanishing on the orbit. Combining \eqref{Theta} with 
\eqref{xII} expresses $\dot\Theta$ as $\dot\Theta(Z,\alpha)$, and the total 
rotation over one shape period is
\begin{equation}\label{DeltaTheta2}
\Delta\Theta(\alpha) \;=\; 2\int_{Z_{2,\alpha}}^{Z_{1,\alpha}}
\frac{g(Z,\alpha)}{\sqrt{(Z_{1,\alpha}-Z)(Z-Z_{2,\alpha})}}\,dZ,
\qquad g(Z,\alpha)=\frac{\dot\Theta(Z,\alpha)}{h(Z,\alpha)}.
\end{equation}
The same trigonometric substitution as in Region I gives
\begin{equation}\label{dthetafinalII}
\Delta\Theta(\alpha) \;=\; 2\int_{-\pi/2}^{\pi/2} 
g\!\left(\tfrac{Z_{1,\alpha}+Z_{2,\alpha}+(Z_{1,\alpha}-Z_{2,\alpha})\sin\phi}{2},\alpha\right)\!d\phi.
\end{equation}

\begin{Prop}\label{prop:analyticII}
The function $\alpha\mapsto\Delta\Theta(\alpha)$ is real-analytic on $(0,\pi/3)$ 
and satisfies $\lim_{\alpha\to 0^+}\Delta\Theta(\alpha)=+\infty$.
\end{Prop}

\begin{proof}
Analyticity: as in Proposition \ref{prop:analyticI}, the argument of the square 
root in $g$ is strictly positive on the integration domain, so the integrand in 
\eqref{dthetafinalII} is jointly analytic in $(\alpha,\phi)$ and the integral 
is analytic in $\alpha$.

Divergence: as $\alpha\to 0^+$ the orbit approaches a heteroclinic connection 
issued from $(0,0,1)$; by Remarks \ref{Rmk1}--\ref{Rmk2}, the period diverges 
while $\dot\Theta$ stays bounded away from zero near the equilibrium, hence 
$\Delta\Theta(\alpha)\to+\infty$.
\end{proof}

\begin{figure}[!h]
\centering
\includegraphics[width=0.5\linewidth]{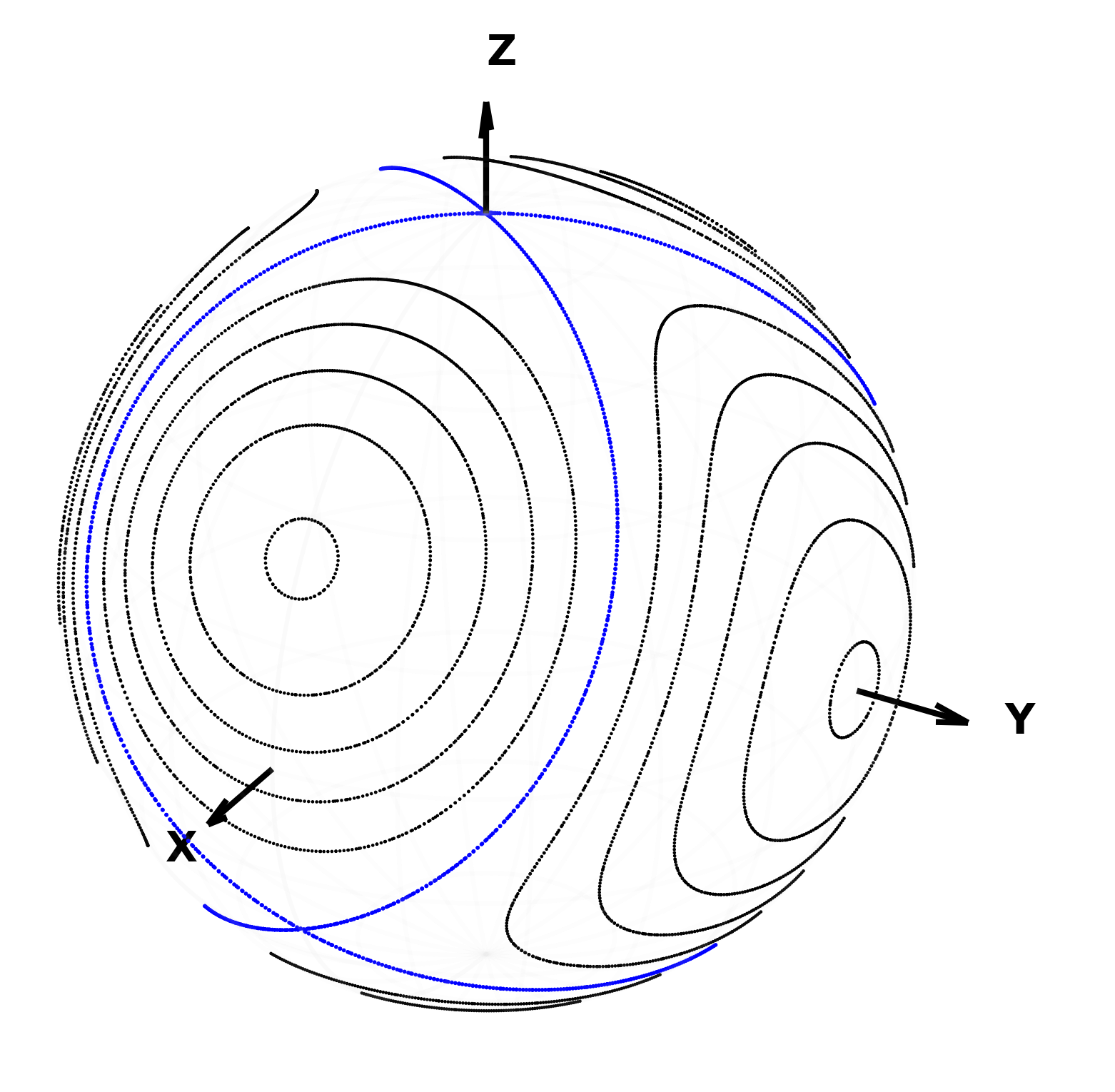}
\caption{The shape sphere with the six heteroclinic orbits highlighted 
(darker shade of blue).}
\label{fig:3}
\end{figure}

\subsection*{Conclusion of the proof of Theorem \ref{thm}}

By Propositions \ref{prop:analyticI} and \ref{prop:analyticII}, in each of the 
two regions the total rotation per period is a real-analytic, non-constant 
function of the orbit parameter. Lemma \ref{lem:generic-irrational} implies 
that $\Delta\Theta/2\pi$ is irrational for a.e.\ orbit in each region. For any 
such orbit the corresponding trajectory of the full system is quasi-periodic 
with two incommensurable frequencies, so its closure is the two-torus obtained 
by rotating the shape orbit; the time average of any smooth observable therefore 
exists and is rotationally invariant. This proves the first assertion of the 
theorem.

The second assertion is immediate: by Propositions \ref{prop:analyticI}--%
\ref{prop:analyticII}, the map $\Delta\Theta$ takes each rational multiple of 
$2\pi$ on a countable set of orbits; on any such orbit the full trajectory is 
periodic and the time average of the empirical density is supported on a finite 
union of closed curves, which is not rotationally invariant.
\qed

\section{Reduction to an ergodic problem on the shape space}\label{sec:4}

In this section we recast Conjecture \ref{congetturan} as a question about a 
measure-preserving dynamical system on the shape space. The key idea is to 
replace the time variable $t$ by a suitably chosen rotation angle $\theta$: 
under mild hypotheses, radiality of the empirical time average becomes 
equivalent to the independence of stroboscopic averages from the initial phase 
along a single trajectory.

\subsection*{Shape-angle reparameterisation}

Fix the centre of vorticity at the origin and the moment of inertia $I$. We 
seek coordinates $(X,\theta)$, with $X$ ranging over a \emph{shape space} 
$\mathcal{S}$ (a compact manifold; for the vortex model with the constraints 
above, $\mathcal{S}\simeq\mathbb{CP}^{N-2}$) and $\theta\in\mathbb{S}^1$, such 
that the full configuration is recovered as
\begin{equation}\label{eq:xRthetaX}
x_i \;=\; R_\theta X_i, \qquad i=1,\dots,N,
\end{equation}
where $R_\theta$ denotes the rotation by $\theta$ about the origin. Concretely, 
$X$ is a particular representative of the equivalence class of $x$ modulo 
rotations, and $\theta$ encodes the rotation that carries $X$ to $x$; different 
choices of the representative differ by the choice of $\theta$.

We impose the following assumptions on this choice.

\begin{assumptions}\label{ass1}
There exist smooth functions $f:\mathcal{S}\to T\mathcal{S}$ and 
$g:\mathcal{S}\to\mathbb{R}$ such that the equations of motion \eqref{eq:hamilton} 
take the form
\begin{equation}\label{eq:shapeangle}
\dot X = f(X), \qquad \dot\theta = g(X),
\end{equation}
and $g$ is bounded away from zero: $g(X)\ge\alpha>0$ for some $\alpha$.
\end{assumptions}

The first two conditions in Assumption \ref{ass1} are automatic, since the 
system is invariant under rotations about the origin: the shape dynamics 
decouples from $\theta$, and $\dot\theta$ is a function of the shape alone. The 
third condition — strict positivity of $\dot\theta$ — is the substantive one and 
depends on the choice of $\theta$; we discuss admissible choices in Section 
\ref{sec:5}. For the rest of this section we work under Assumption \ref{ass1}.

Positivity of $\dot\theta$ allows us to use $\theta$ as independent variable. 
Denoting by $'$ the derivative with respect to $\theta$, 
\eqref{eq:shapeangle} becomes
\begin{equation}\label{time2angle}
X' \;=\; \frac{dX}{d\theta} \;=\; \frac{\dot X}{\dot\theta} 
\;=\; \frac{f(X)}{g(X)} \;=:\; F(X),
\end{equation}
so the shape trajectory as a function of $\theta$ is governed by the autonomous 
system $X'=F(X)$ on $\mathcal{S}$. The flow \eqref{time2angle} preserves the 
measure with density $g$ times the natural measure inherited from the 
Hamiltonian flow, and is in particular measure-preserving on $\mathcal{S}$.

\subsection*{From time averages to angle averages}

For any smooth bounded observable $\Psi:\mathbb{R}^{2N}\to\mathbb{R}$, the time 
average along a trajectory $x(t)=R_{\theta(t)}X(\theta(t))$ can be rewritten as 
an angle average via the change of variable $t\mapsto\theta$:
\begin{equation}\label{eq:timetoangle}
\langle\Psi\rangle \;=\; \lim_{T\to\infty}\frac{1}{T}\int_0^T \Psi(x(t))\,dt
\;=\; \lim_{\Theta\to\infty}\frac{1}{t(\Theta)}\int_{\theta_0}^{\Theta}
\frac{\Psi(R_\theta X(\theta))}{g(X(\theta))}\,d\theta,
\end{equation}
where $t(\theta)$ is the inverse of $\theta(t)$ and $X(\theta):=X(t(\theta))$. 
Since $\dot\theta\ge\alpha>0$, one has $t(\Theta)\sim\Theta/\langle g^{-1}\rangle^{-1}$ 
as $\Theta\to\infty$, and \eqref{eq:timetoangle} is well posed.

The right-hand side of \eqref{eq:timetoangle} is an average, along the 
trajectory of \eqref{time2angle}, of the $\theta$-dependent observable
\[
\Xi(\theta,X) \;=\; \frac{\Psi(R_\theta X)}{g(X)}.
\]
Radiality of the empirical time average — the property that 
$\langle\Psi\circ R_\delta\rangle=\langle\Psi\rangle$ for every 
$\delta\in[0,2\pi)$ — translates into the following statement about $\Xi$.

\begin{thm}\label{rotinv}
Let $\Psi:\mathbb{R}^{2N}\to\mathbb{R}$ be continuous and bounded, and suppose 
Assumption \ref{ass1} holds. Consider the flow \eqref{time2angle} on 
$\mathcal{S}$ and, for each $X_0\in\mathcal{S}$ and each initial phase 
$\theta_0\in\mathbb{S}^1$, the stroboscopic sequence 
$X_k^{\theta_0}:=X(\theta_0+2\pi k)$, $k\in\mathbb{N}$. Assume that for the 
observable $\Xi(\theta,X)=\Psi(R_\theta X)/g(X)$ the stroboscopic average
\begin{equation}\label{eq:stroboscopic}
\bar\Xi(X_0,\theta_0) \;=\; \lim_{N\to\infty}\frac{1}{N}\sum_{k=0}^{N-1}
\Xi\!\big(\theta_0+2\pi k,\,X_k^{\theta_0}\big)
\end{equation}
exists and is independent of $\theta_0$ for a.e.\ $X_0\in\mathcal{S}$. Then the 
continuous time average $\langle\Psi\rangle$ exists and is rotationally invariant, 
i.e.\ $\langle\Psi\circ R_\delta\rangle(x_0)=\langle\Psi\rangle(x_0)$ for every 
$\delta\in[0,2\pi)$ and a.e.\ $x_0$.
\end{thm}

The proof is a direct consequence of the change of variables 
\eqref{eq:timetoangle} and is given in Appendix \ref{app:A}.

\subsection*{The underlying ergodic problem}

Theorem \ref{rotinv} reduces Conjecture \ref{congetturan} to a question about a 
measure-preserving dynamical system on $\mathcal{S}$, which we formulate in 
general terms below since it appears to be of independent interest.

\begin{question}\label{ergodic-question}
Let $(\Lambda,\mu,\Phi_t)$ be a measure-preserving dynamical system on a 
compact manifold $\Lambda$. Fix $\tau>0$. By Birkhoff's theorem, for every 
$\tau_0\in[0,\tau)$ and every smooth observable $\Psi$ the stroboscopic 
average
\[
\langle\Psi\rangle_{\tau_0}(x) \;=\; \lim_{n\to\infty}
\frac{1}{n}\sum_{k=0}^{n-1}\Psi\!\big(\Phi_{\tau_0+k\tau}(x)\big)
\]
exists for $\mu$-a.e.\ $x\in\Lambda$. Under which conditions on the flow $\Phi_t$ 
is $\langle\Psi\rangle_{\tau_0}$ independent of $\tau_0$, $\mu$-a.e.?
\end{question}

The condition seems very natural for non-degenerate systems, but it is not 
automatic.

\begin{Rmk}\label{rmk:harmonic}
A simple counterexample is the harmonic oscillator of period $\tau$: for every 
$\tau_0$, $\Phi_{\tau_0+k\tau}(x)=\Phi_{\tau_0}(x)$, so 
$\langle\Psi\rangle_{\tau_0}(x)=\Psi(\Phi_{\tau_0}(x))$ depends non-trivially 
on $\tau_0$. The same phenomenon appears in a physically important case: applying 
the Binet equation to the Kepler problem \cite{GPS}, the angle equation is that 
of a harmonic oscillator, and consistently all bounded Kepler orbits are closed 
and the corresponding measures are not rotationally invariant. More generally, 
systems that are canonically equivalent to a collection of harmonic oscillators 
with frequencies commensurable to $\tau$ provide analogous counterexamples.
\end{Rmk}

It seems difficult to construct counterexamples outside this class of 
harmonic-oscillator-like systems, and one may hope that some non-degeneracy 
condition on the spectrum of the linearised flow — much weaker than a complete 
characterisation of the dynamics — is enough to guarantee $\tau_0$-independence. 
We do not settle this here; we record Problem \ref{ergodic-question} as, in our 
view, the natural abstract framework for Conjecture \ref{congetturan}.

\section{Applications to specific vortex systems}\label{sec:5}

We conclude by showing that the setup of Section \ref{sec:4} applies in two 
concrete situations: it recovers the three-vortex result of Section \ref{sec:3} 
and it verifies Assumption \ref{ass1} in the KAM regime of the four-vortex 
problem, giving a partial reduction of Conjecture \ref{congetturan} to Problem 
\ref{ergodic-question} in that regime.

\subsection{Three vortices, revisited}\label{sec:5.1}

For $N=3$, restrict attention to a region of the shape sphere in which the 
vortices do not cross the axis $Y=0$ (Region I of Section \ref{sec:3}). 
Choosing shape variables $(X,Y,Z)$ as in \eqref{eq:XYZ} and the angle
\[
\theta \;=\; \frac{\Theta_1+\Theta_2+\Theta_3}{3},
\]
Assumption \ref{ass1} is satisfied: by Remark \ref{Rmk1} the function 
$\dot\theta$ given by \eqref{Theta} is smooth and strictly positive throughout 
this region. The stroboscopic sequence $X(\theta_0+2\pi k)$ is then a rigid 
rotation of the shape orbit on the shape sphere by the angle 
$\Delta\Theta(R)\bmod 2\pi$, where $\Delta\Theta(R)$ is the total rotation per 
shape period computed in \eqref{dthetafinal}. By Proposition 
\ref{prop:analyticI} and Lemma \ref{lem:generic-irrational}, this angle is 
irrational for a.e.\ orbit; the corresponding stroboscopic sequence is therefore 
equidistributed on the closed shape orbit, and Theorem \ref{rotinv} recovers the 
first assertion of Theorem \ref{thm} in Region I. Region II is treated 
analogously.

\subsection{Four vortices in the Khanin regime}\label{sec:5.2}

The four-vortex system is not integrable in general, and regions of chaotic 
motion are known to exist \cite{Z,AP82}. Khanin \cite{Kh} proved the existence 
of a positive-measure set of initial data on which the motion is quasi-periodic, 
corresponding to a hierarchical configuration of two tight pairs: 
$|z_1-z_2|<\varepsilon$, $|z_3-z_4|<\varepsilon$, with the two pair centres of 
vorticity separated by a distance of order one. In this regime the system can 
be viewed as a small perturbation of two independent binary systems.

We verify Assumption \ref{ass1} in this regime. Let 
$z_A=\tfrac{1}{2}(z_1+z_2)$ and $z_B=\tfrac{1}{2}(z_3+z_4)$ denote the pair 
centres, and let $\theta$ be the orientation of the vector $z_B-z_A$ with 
respect to a fixed reference frame. By the KAM non-degeneracy conditions 
verified in \cite{Kh}, the invariant tori persist and the pairs remain tight; 
we show that $\dot\theta>0$ on this set.

\begin{figure}[!htbp]
\centering
\includegraphics[width=0.45\linewidth]{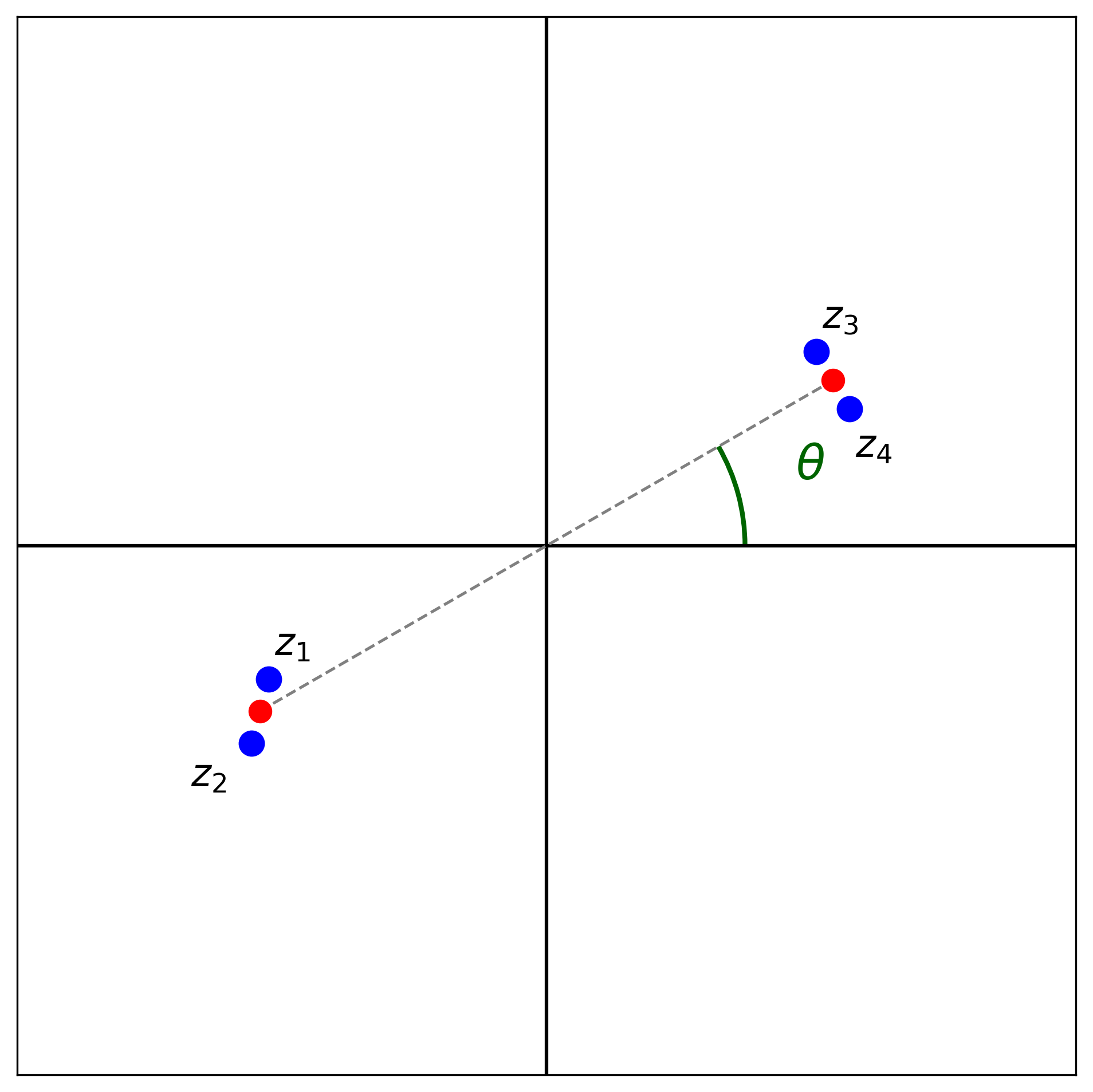}
\caption{The four-vortex configuration in the Khanin regime: two tight pairs 
with inter-pair separation of order one.}
\end{figure}

By rotational invariance we may fix a reference frame in which 
$z_A=(-L,0)$ and $z_B=(L,0)$ for some $L>0$. Since the internal interaction 
between $z_3$ and $z_4$ does not contribute to the velocity of their centre, we 
have
\[
\dot y_B \;=\; \tfrac{1}{2}(\dot y_3+\dot y_4) 
\;=\; \frac{\Omega}{16\pi}\sum_{i\in\{1,2\}}\sum_{j\in\{3,4\}}
\frac{x_j-x_i}{|z_i-z_j|^2},
\]
and similarly for $\dot y_A$ with the opposite sign convention. Under the tight-%
pair assumption $x_i\in[-L-\varepsilon,-L+\varepsilon]$ and 
$x_j\in[L-\varepsilon,L+\varepsilon]$, so 
$x_j-x_i\ge 2(L-\varepsilon)>0$ for every $\varepsilon<L$. Hence 
$\dot y_B>0$ and $\dot y_A<0$, which gives $\dot\theta>0$; a lower bound 
$\dot\theta\ge\alpha>0$ follows from the same estimate together with the 
uniform bounds $|z_i-z_j|\le 2L+2\varepsilon$. 

Assumption \ref{ass1} is 
satisfied on the Khanin set, and Conjecture \ref{congetturan} in this regime 
reduces to Problem \ref{ergodic-question} for the induced flow on the reduced 
shape space. In particular we expect that non-degeneracy of the frequencies n the KAM tori would imply the result.

In the chaotic regions of the four-vortex phase space we expect the same 
picture to hold on heuristic grounds — decay of time correlations should make 
the shape decorrelate from the rotation angle — but we do not pursue this 
direction here.

Conjecture~\ref{congetturan} addresses the general measure-theoretic behavior of $N$ vortices.
In the case of four vortices its mechanism can be verified explicitly in low-dimensional invariant submanifolds. The following remark details the four-vortex symmetric case, where the validity of the conjecture can be established analytically by quadrature.

\begin{Rmk}[The symmetric four-vortex submanifold: an exactly solvable case of Conjecture~\ref{congetturan}]\label{rmk:sym4}
For $N=4$, the submanifold
\[\mathcal M \;=\; \bigl\{z\in\Gamma_4:\ z_3=-z_1,\ z_4=-z_2\bigr\}\]
is invariant under the flow of \eqref{VM}. On $\mathcal M$, the system reduces by energy conservation to a one-dimensional periodic shape oscillator, governed by the conserved quantity
\[E \;:=\; \prod_{1\le i<j\le 4}\vert{}z_i-z_j\vert{} \;=\; e^{-8\pi H/\Omega}.\]
After one period of the shape oscillation, the configuration returns to its initial shape rotated by a physical angle $\Delta\theta_E$, given in closed form by the complete Abelian integral
\[\Delta\theta_E \;=\; E\!\int_{u_-}^{u_+}\! \left[\frac{16\,(u-1)\,\sqrt{u(2-u)}}{u\,E}       +\frac{3}{2u}\right] \frac{du}{\sqrt{p(z(u))}},\]
where $u=\vert{}z_1\vert{}^2 \in [u_-, u_+]$ delimit the shape oscillation interval, $z(u)=\sqrt{u(2-u)}$, and $p(z)=(16z-E)(16z^3-16z+E)$.

Since $\Delta\theta_E$ is a real-analytic, non-constant function of $E$ on $(\frac{32\sqrt{3}}{9}, 16)$, the rotation angle $\Delta\theta_E/(2\pi)$ is an irrational number for almost every energy $E$. By Theorem~\ref{rotinv}, this implies that for almost every initial datum on $\mathcal M$, the time-averaged empirical vorticity measure exists and is invariant under rotations, settling Conjecture~\ref{congetturan} completely on this invariant submanifold. The detailed proof and phase portrait analysis are deferred to Appendix~\ref{app:sym4_proof}.
\end{Rmk}

\begin{figure}[!htbp]
\centering
\includegraphics[width=0.45\linewidth]{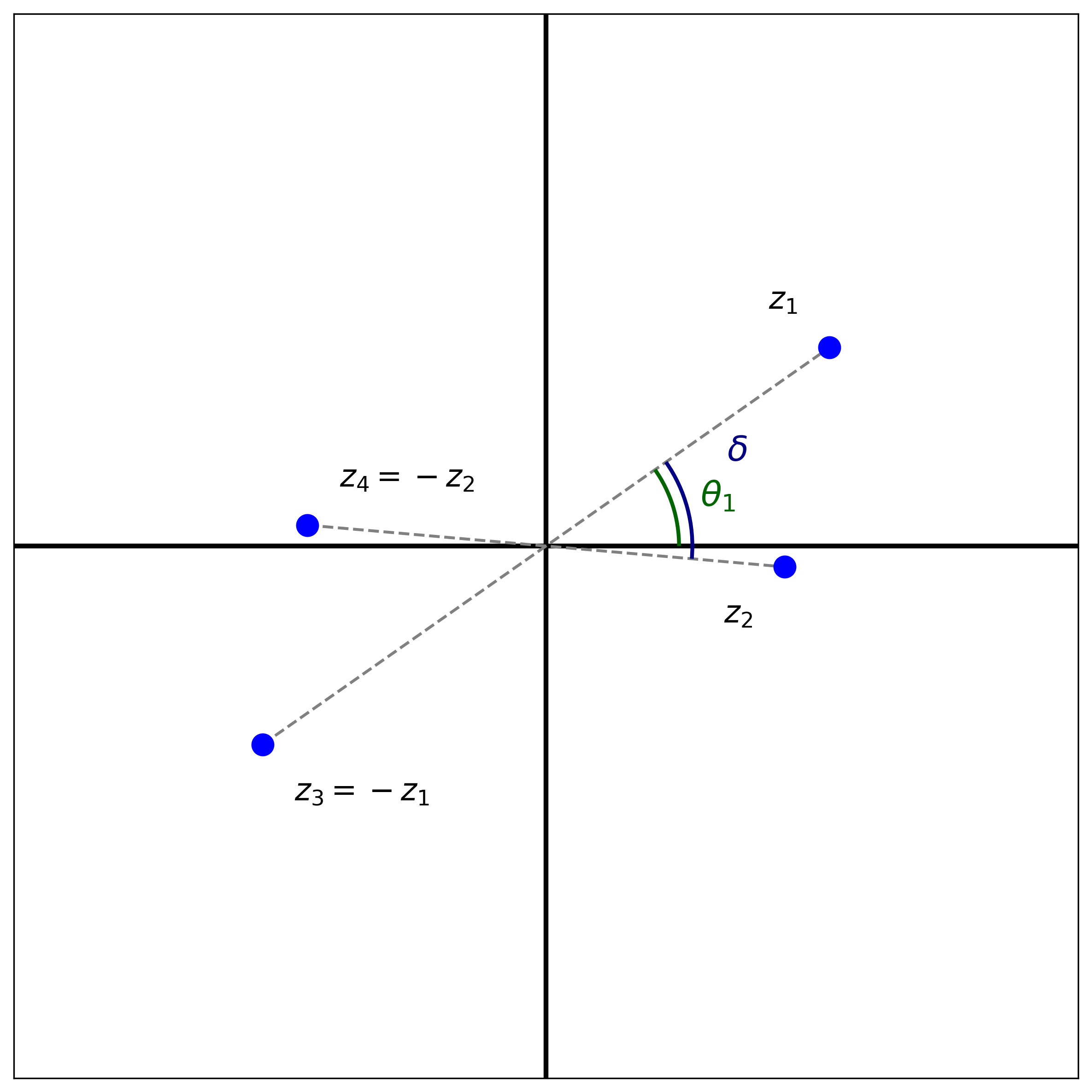}
\caption{Four symmetric vortices.}
\label{fig:4symm}
\end{figure}

\subsection*{Use of Generative AI in the Writing Process}
During the preparation of this work, the authors used ChatGPT and Claude to 
improve the language and readability of the manuscript. After using these 
tools, the authors reviewed and edited the content as needed and take full 
responsibility for the framework and final output of the publication.

\newpage

\begin{appendices}
\section{Proof of Theorem \ref{rotinv}}\label{app:A}

\begin{proof}[Proof of Theorem \ref{rotinv}]
By the change of variable $t\mapsto\theta$ described in \eqref{eq:timetoangle}, 
the continuous time average of $\Psi$ along the trajectory issued from $x_0$ 
can be written as
\begin{equation}\label{eq:app-1}
\langle\Psi\rangle(x_0) \;=\; \lim_{\Theta\to\infty}\frac{1}{t(\Theta)}
\int_{\theta_0}^{\Theta}\Xi(\theta,X(\theta))\,d\theta,
\qquad 
\Xi(\theta,X) \;=\; \frac{\Psi(R_\theta X)}{g(X)},
\end{equation}
where $X(\theta)$ is the shape trajectory of \eqref{time2angle} issued from 
$X_0$ and $\theta_0$ is the initial phase.

Fix $\theta_0=0$ and take $\Theta=2\pi n$ with $n\in\mathbb{N}$. Partitioning 
the integration domain into $n$ intervals of length $2\pi$ and changing 
variables in each interval via $\theta=2\pi k+\vartheta$, 
$\vartheta\in[0,2\pi)$, we get
\begin{equation}\label{eq:app-2}
\int_{0}^{2\pi n}\Xi(\theta,X(\theta))\,d\theta
\;=\; \sum_{k=0}^{n-1}\int_{0}^{2\pi}
\Xi\!\big(\vartheta+2\pi k,\,X(\vartheta+2\pi k)\big)\,d\vartheta.
\end{equation}
Since $\Xi$ is $2\pi$-periodic in its first argument (it depends on $\theta$ 
only through $R_\theta$), this rewrites as
\begin{equation}\label{eq:app-3}
\int_{0}^{2\pi n}\Xi(\theta,X(\theta))\,d\theta
\;=\; \int_{0}^{2\pi}\!d\vartheta \sum_{k=0}^{n-1}
\Xi\!\big(\vartheta,\,X(\vartheta+2\pi k)\big).
\end{equation}
Dividing by $n$ and passing to the limit, the inner sum is precisely the 
stroboscopic average \eqref{eq:stroboscopic} with initial phase $\vartheta$. By 
hypothesis this limit exists for a.e.\ $X_0$ and does not depend on 
$\vartheta$: denoting the common value by $\bar\Xi(X_0)$, dominated 
convergence gives
\begin{equation}\label{eq:app-4}
\lim_{n\to\infty}\frac{1}{2\pi n}\int_{0}^{2\pi n}\Xi(\theta,X(\theta))\,d\theta
\;=\; \bar\Xi(X_0).
\end{equation}
Combining \eqref{eq:app-1} and \eqref{eq:app-4}, and using 
$t(2\pi n)/(2\pi n)\to\langle 1/g\rangle$ (which is finite and independent of 
$x_0$ a.e.\ by the assumed independence of stroboscopic averages applied to 
$\Xi\equiv 1/g$), we conclude
\[
\langle\Psi\rangle(x_0) \;=\; \frac{\bar\Xi(X_0)}{\langle 1/g\rangle}.
\]
The right-hand side depends on $x_0$ only through the shape $X_0$: rotating 
$x_0$ by $\delta$ replaces $\theta_0$ by $\theta_0+\delta$, which leaves $X_0$ 
unchanged. Hence $\langle\Psi\rangle(R_\delta x_0)=\langle\Psi\rangle(x_0)$ 
for every $\delta\in[0,2\pi)$, which is the claimed rotational invariance.
\end{proof}

\section{Symmetric Four-Vortex Case}\label{app:sym4_proof}

In this appendix we provide the full derivation and the proof sketched in Remark~\ref{rmk:sym4}.


To state the results with constants independent of $\Omega$, set $E := \prod_{1\le i<j\le 4}|z_i-z_j| = e^{-8\pi H/\Omega}$, which is conserved along the flow. All statements follow from quadrature of the reduced equations.

\smallskip
\noindent\textbf{(a) Reduction to one degree of freedom.}
On $\mathcal M$ write $a_1=z_1$, $a_2=z_2$; the constraint $I=1$ reads $|a_1|^2+|a_2|^2=2$. Introducing
\[
u=|a_1|^2\in[0,2],\qquad \delta=\arg a_1-\arg a_2 ,
\]
the conserved quantity $E$ becomes
\[
E=16\sqrt{u(2-u)}\,\bigl(1-u(2-u)\cos^2\delta\bigr),
\]
so that $c^2(u,\delta):=u(2-u)\cos^2\delta=1-\dfrac{E}{16\sqrt{u(2-u)}}$ is a function of $u$ alone along each orbit. The evolution of $u$ is determined by energy conservation:
\[
\dot u^{\,2}=\frac{4}{E^{2}}\,p\bigl(z(u)\bigr),\qquad z(u):=\sqrt{u(2-u)},
\]
with the quartic polynomial $p(z)=(16z-E)\,(16z^{3}-16z+E)$.
Hence $u$ oscillates periodically between two consecutive roots of $p(z(u))$ enclosing a region where $p>0$, exactly as a particle in a one-dimensional potential (in the reduced time $s=+\frac{\Omega}{8\pi}t$).

\smallskip
\noindent\textbf{(b) Phase portrait.} The structure of the roots of $p$ is governed by two thresholds:
\begin{itemize}
\item \emph{Separatrix value $E=32\sqrt3/9\approx 6.1584$:} Here $p(z)$ has a double root at $z=1/\sqrt3$, corresponding to two collinear relative equilibria of saddle type.
\item \emph{Case~3, $E<32\sqrt3/9$:} $p(z(u))>0$ on three intervals: one ``symmetric'' interval enclosing $u=1$, and two ``asymmetric'' intervals close to $u=0$ and $u=2$.
\item \emph{Case~2, $E=32\sqrt3/9$:} Heteroclinic connections between the two collinear equilibria.
\item \emph{Case~1, $32\sqrt3/9<E<16$:} A single oscillation interval symmetric under $u\mapsto 2-u$. As $E\uparrow 16$, it shrinks to the square equilibrium ($u=1, \delta=\pi/2$) rotating with constant angular velocity $\omega=3/2$ in reduced time $s$.
\end{itemize}

\smallskip
\noindent\textbf{(c) Period and rotation numbers in closed form.}
Let $[u_-, u_+]$ delimit one oscillation interval. The period $T_E$ in reduced time and the increments $\Delta\delta_E$, $\Delta\theta_E$ over one full oscillation of $u$ are given by integrating $dt = du/\dot u$, $\dot\delta/\dot u$, and $\dot\theta_1/\dot u$:
\[T_E \;=\; E\!\int_{u_-}^{u_+}\frac{du}{\sqrt{p(z(u))}}\,,\]
\[\Delta\delta_E = E\!\int_{u_-}^{u_+}\! \frac{(u-1)\,\bigl(3c^2(u)-1\bigr)}      {u(2-u)\,\bigl(1-c^2(u)\bigr)}\, \frac{du}{\sqrt{p(z(u))}},\]
\[\Delta\theta_E = E\!\int_{u_-}^{u_+}\! \left[\frac{16\,(u-1)\,\sqrt{u(2-u)}}{u\,E}       +\frac{3}{2u}\right] \frac{du}{\sqrt{p(z(u))}} .\]

\smallskip
\noindent\textbf{(d) Topological nature of $\Delta\delta\_E$.}
The involution $u\mapsto 2-u$ changes the sign of the numerator of $\Delta\delta_E$ while leaving the denominator and $p(z(u))$ invariant, so $\Delta\delta_E\equiv 0$ on every symmetric interval (Case 1 and middle interval of Case 3). On asymmetric intervals, $\Delta\delta_E=\pm\pi$, forced by the topology of the shape space ($\delta$ defined mod $\pi$ due to $2\leftrightarrow 4$ vortex swapping). Thus, the shape orbit closes after \emph{one} period of $u$-oscillation in all cases.

\smallskip
\noindent\textbf{(e) Consequence for Conjecture~\ref{congetturan}.}
The dynamics on $\mathcal M$ is a skew product over a periodic shape oscillator. The time-averaged empirical measure is radial if and only if $\Delta\theta_E/(2\pi)\in\mathbb{R}\setminus\mathbb{Q}$. Since $\Delta\theta_E$ is real-analytic and non-constant in $E$, it diverges logarithmically at $E\downarrow 32\sqrt3/9$ and approaches $3\pi/\sqrt2$ at the square equilibrium. The set of resonant (rational) energies is at most countable, proving that the time-averaged measure exists and is radial for almost every initial datum on $\mathcal M$.

\end{appendices}

\end{document}